\documentclass[a4paper]{llncs}

\usepackage{amssymb}
\usepackage{pifont}
\usepackage{listings}
\usepackage{graphicx}
\usepackage{xspace}
\usepackage{pgf}
\usepackage[noend]{algpseudocode}
\usepackage{algorithm}
\usepackage{amsfonts}
\usepackage{amsmath}
\usepackage{pifont}
\usepackage{multicol}
\usepackage[justification=centering]{caption}
\usepackage{subcaption}
\usepackage{subfig}
\usepackage{mdframed,framed}
\usepackage{wasysym,stmaryrd}
\usepackage{xcolor}
\usepackage{mathrsfs}
\usepackage{tikz}
\usetikzlibrary{arrows, automata, shapes, calc, positioning, matrix}

\algrenewcommand\algorithmicrequire{\textbf{Precondition:}}
\algrenewcommand\algorithmicensure{\textbf{Postcondition:}}

\newcommand*\Let[2]{\State #1 $\gets$ #2}
\newcommand{\sm}[1]{{\small \mbox{$#1$}}}

\newcommand{\defeq}{\ensuremath{\stackrel{\mathrm{def}}{=}}}

\newcommand{\hassume}{\mathit{assume}}
\newcommand{\hassert}{\mathit{assert}}

\newcommand{\hnew}{\mathit{new}}

\newcommand{\hif}{\mathit{if}}
\newcommand{\helse}{\mathit{else}}

\newcommand{\hwhile}{\mathit{while}}
\newcommand{\hnull}{\ensuremath{\mathbf{null}}}

\newcommand{\nil}{{\bf null}}

\newcommand{\path}{\mathit{path}}

\newcommand{\sem}[1]{\ensuremath{\llbracket{#1}\rrbracket}}

\newcommand{\hpathlen}{\ensuremath{\mathit{pathLength}}\xspace}
\newcommand{\halloc}{\ensuremath{\mathit{new}}\xspace}
\newcommand{\hassign}{\ensuremath{\mathit{assign}}\xspace}
\newcommand{\hlookup}{\ensuremath{\mathit{lookup}}\xspace}
\newcommand{\hupdate}{\ensuremath{\mathit{update}}\xspace}
\newcommand{\halias}{\ensuremath{\mathit{alias}}\xspace}
\newcommand{\hisnull}{\ensuremath{\mathit{isNull}}\xspace}
\newcommand{\hispath}{\ensuremath{\mathit{isPath}}\xspace}
\newcommand{\hcircular}{\ensuremath{\mathit{circular}}\xspace}
\newcommand{\id}{\ensuremath{\mathrm{id}}\xspace}
\newcommand{\hsubdivide}{\ensuremath{\mathit{subdivide}}\xspace}
\newcommand{\hfresh}{\ensuremath{\mathit{fresh}()}\xspace}

\newcommand{\shakira}{{\sc Shakira}\xspace}

\title{Propositional Reasoning about Safety and Termination of Heap-Manipulating Programs} 
\author{Cristina David \and Daniel Kroening \and Matt Lewis}
\institute{University of Oxford}

\newenvironment{keywords}{
       \list{}{\advance\topsep by0.35cm\relax\small
       \leftmargin=0cm
       \labelwidth=0.35cm
       \listparindent=0.35cm
       \itemindent\listparindent
       \rightmargin\leftmargin}\item[\hskip\labelsep
                                     \bfseries Keywords:]}
     {\endlist}

\makeatletter
\pgfdeclareshape{datastore}{
  \inheritsavedanchors[from=rectangle]
  \inheritanchorborder[from=rectangle]
  \inheritanchor[from=rectangle]{center}
  \inheritanchor[from=rectangle]{base}
  \inheritanchor[from=rectangle]{north}
  \inheritanchor[from=rectangle]{north east}
  \inheritanchor[from=rectangle]{east}
  \inheritanchor[from=rectangle]{south east}
  \inheritanchor[from=rectangle]{south}
  \inheritanchor[from=rectangle]{south west}
  \inheritanchor[from=rectangle]{west}
  \inheritanchor[from=rectangle]{north west}
  \backgroundpath{
    \southwest \pgf@xa=\pgf@x \pgf@ya=\pgf@y
    \northeast \pgf@xb=\pgf@x \pgf@yb=\pgf@y
    \pgfpathmoveto{\pgfpoint{\pgf@xa}{\pgf@ya}}
    \pgfpathlineto{\pgfpoint{\pgf@xb}{\pgf@ya}}
    \pgfpathmoveto{\pgfpoint{\pgf@xa}{\pgf@yb}}
\usepackage{tikz}
\usetikzlibrary{arrows, automata, shapes}

    \pgfpathlineto{\pgfpoint{\pgf@xb}{\pgf@yb}}
 }
}
\tikzset{>=stealth}
\makeatother

\begin{document}
\maketitle
\pagestyle{headings}  

\begin{abstract}
This paper shows that it is possible to reason about the safety and termination of programs handling potentially  
cyclic, singly-linked lists using propositional reasoning even when the safety invariants and termination arguments 
depend on constraints over the lengths of lists.  
For this purpose, we propose the theory SLH of singly-linked lists with length, which is
able to capture non-trivial interactions between shape and arithmetic.
When using the theory of bit-vector arithmetic as a background,
SLH is efficiently decidable via a reduction to SAT.
We show the utility of SLH for software verification by using it to 
express safety invariants and termination arguments for programs 
manipulating potentially cyclic, singly-linked lists with unrestricted, unspecified sharing.
We also provide an implementation of the decision procedure
 and use it to check safety and termination proofs for several heap-manipulating programs.
\end{abstract}


\section{Introduction}

Proving safety of heap-manipulating programs is a notoriously difficult task. 
One of the main culprits is the complexity of the verification conditions
generated for such programs.  The constraints comprising these verification
conditions can be arithmetic (e.g.  the value stored at location pointed by
$x$ is equal to 3), structural (e.g.  $x$ points to an acyclic singly-linked
list), or a combination of the first two when certain structural properties
of a data structure are captured as numeric values (e.g.  the length of the
list pointed by $x$ is 3).  Solving these combined constraints
requires non-trivial interaction between shape and arithmetic.



For illustration, consider the program in Figure~\ref{fig:motivation.c},
which iterates simultaneously over the lists $x$ and $y$.  The program is
safe, i.e., there is no null pointer dereferencing and the assertion after
the loop holds.  While the absence of null pointer dereferences is trivial
to observe and prove, the fact that the assertion after the loop holds
relies on the fact that at the beginning of the program and after each loop
iteration the lengths of the lists $z$ and $t$ are equal.  Thus,
the specification language must be capable of expressing the fact that both
$z$ and $t$ reach null in the same number of steps.  Note that the
interaction between shape and arithmetic constraints is intricate, and
cannot be solved by a mere theory combination.

The problem is even more pronounced when proving termination of
heap-manipulating programs.  The reason is that, even more frequently than
in the case of safety checking, termination arguments depend on the size of
the heap data structures.  For example, a loop iterating over the nodes of
such a data structure terminates after all the reachable nodes have been
explored.  Thus, the termination argument is directly linked to the number
of nodes in the data structure.  This situation is illustrated again by the
loop in Figure~\ref{fig:motivation.c}.


There are few logics capable of expressing this type of
interdependent shape and arithmetic constraint.  One of the reasons is
that, given the complexity of the constraints, such logics can easily become
undecidable (even the simplest use of transitive closure leads to
undecidability \cite{DBLP:conf/csl/ImmermanRRSY04}), or at best inefficient.

The tricky part is identifying a logic expressive enough to capture the
corresponding constraints while still being efficiently decidable.  One work
that inspired us in this endeavor is the recent approach by Itzhaky et al.  on
reasoning about reachability between dynamically allocated memory locations
in linked lists using effectively-propositional (EPR)
reasoning~\cite{DBLP:conf/cav/ItzhakyBINS13}.  This result is appealing as
it can harness advances in SAT solvers.  The only downside is that the
logic presented in~\cite{DBLP:conf/cav/ItzhakyBINS13} is better suited for
safety than termination checking, and is best for situations where
safety does not depend on the interaction between shape and arithmetic. 
Thus, our goal is to define a logic that can be used in such scenarios
while still being reducible to SAT.

This paper shows that it is possible to reason about the safety and
termination of programs handling potentially cyclic, singly-linked lists
using propositional reasoning.  For this purpose, we present the logic SLH
which can express interdependent shape and arithmetic constraints.  We
empirically prove its utility for the verification of heap-manipulating
programs by using it to express safety invariants and termination arguments
for intricate programs with potentially cyclic, singly-linked lists with
unrestricted, unspecified sharing.

SLH is parametrised by the background arithmetic theory used to express the
length of lists (and implicitly every numeric variable).  The decision
procedure reduces validity of a formula in SLH to satisfiability of a
formula in the background theory.  Thus, SLH is decidable if the background
theory is decidable.

As we are interested in a reduction to SAT, we instantiate SLH with the
theory of bit-vector arithmetic, resulting in SLH[$\mathcal{T_{BV}}$].  This
allows us to handle non-linear operations on lists length (e.g.  the example
in Figure~\ref{fig:motivation.d}), while still retaining decidability. 
However, SLH can be combined with other background theories, e.g.,
Presburger arithmetic.

We provide an implementation of our decision procedure for
SLH[$\mathcal{T_{BV}}$] and test its efficiency by verifying a suite of
programs against safety and termination specifications expressed in SLH. 
Whenever the verification fails, our decision procedure produces a
counterexample.




\begin{framed}
\emph{Contributions:} \\
\noindent * We propose the theory SLH of singly-linked lists with length. 
SLH allows \emph{unrestricted} sharing and cycles.

\noindent * We define the strongest post-condition for formulae in SLH.

\noindent * We show the utility of SLH for software verification by using it to 
express safety invariants and termination arguments for programs 
with potentially cyclic singly-linked lists. 

\noindent * We present the  instantiation  SLH[$\mathcal{T_{BV}}$] of SLH with the theory of bit-vector arithmetic.
SLH[$\mathcal{T_{BV}}$] can express non-linear operations on the lengths of lists, while still retaining decidability. 

\noindent * We provide a reduction from satisfiability of SLH[$\mathcal{T_{BV}}$] to propositional SAT.



\noindent * We provide an implementation of the decision procedure for SLH[$\mathcal{T_{BV}}$] and test it by checking safety and termination for several heap-manipulating programs 
(against provided safety invariants and termination arguments).

\end{framed}

\section{Motivation} \label{sec:motivation}

\begin{figure}
\begin{tabular}{cc}
\begin{subfigure}[b]{0.5\textwidth}
\begin{lstlisting}[mathescape=true,flexiblecolumns=true]
  List x, y = x;
  int n = length(x), i = 0;

  while (i < n) {
    y = y$\rightarrow$next;
    i = i+1;
  }
\end{lstlisting}
\caption{}
 \label{fig:motivation.a}
\end{subfigure}%
&
\begin{subfigure}[b]{0.5\textwidth}
\begin{lstlisting}[mathescape=true,flexiblecolumns=true]
  List x, y, z = x, t = y;

  assume(length(x)==length(y));

  while (z!=NULL && t!=NULL) {
    z = z$\rightarrow$next;
    t = t$\rightarrow$next;
  }

  assert (z==NULL && t==NULL);
\end{lstlisting}
\caption{}
 \label{fig:motivation.c}
\end{subfigure}%
\\ \\ \\ \\ \\ \\ \\
\begin{subfigure}[b]{0.6\textwidth}
\begin{lstlisting}[mathescape=true,flexiblecolumns=true]
int divides(List x, List y) {
  List z = y;
  List w = x;

  assume(length(x)!=MAXINT && 
         length(y)!=MAXINT &&
         y != NULL);
       
  while (w != NULL) {
    if (z == NULL) z = y;
    z = z$\rightarrow$next; 
    w = w$\rightarrow$next;
  }

  assert(z == NULL $\Leftrightarrow$ 
         length(x)$\%$length(y)==0);
  return z == NULL;
}
\end{lstlisting}
\caption{}
 \label{fig:motivation.d}
\end{subfigure}%
&
\begin{subfigure}[b]{0.4\textwidth}
\begin{lstlisting}[mathescape=true,flexiblecolumns=true]
int isCircular(List l) {
  List p = q = l;
  
  do {
    if (p != NULL) p = p$\rightarrow$next;
    if (q != NULL) q = q$\rightarrow$next;
    if (q != NULL) q = q$\rightarrow$next;
  } 
  while (p!=NULL && q!=NULL 
         && p!=q);

  assert(p == q $\Leftrightarrow$ circular(l));
  return p == q;
}

\end{lstlisting}
\caption{}
\label{fig:motivation.e}
\end{subfigure}%
\end{tabular}
\caption{Motivational examples. \label{fig:motivation}}
\end{figure}

Consider the examples in Figure~\ref{fig:motivation}. They all capture situations where the safety  
(i.e. absence of null pointer dereferencing and no assertion failure) and termination of the program
depend on interdependent shape and arithmetic constraints. 
In this section we only give an intuitive description of these examples, and we 
revisit and formally specify them in Section~\ref{sec:revisit.motivation}. 
We assume the existence of the following two functions:
(1)~$length(x)$ returns the number of nodes on the path from $x$ to NULL if
the list pointed by $x$ is acyclic, and MAXINT otherwise; (2)~$circular(x)$
returns true iff the list pointed by $x$ is circular (i.e., $x$ is part of a
cycle).

In Figure~\ref{fig:motivation.a}, we iterate over the potentially cyclic 
singly-linked list pointed by \sm{x} a number of times equal with the result of $length(x)$. 
The program is safe (i.e. $y$ is not NULL at loop entry) and terminating.  
A safety invariant for the loop  
needs to capture the length of the path from $y$ to NULL.

The loop in Figure~\ref{fig:motivation.c} iterates over the lists pointed by $x$ and $y$, respectively, 
until one of them becomes NULL.
In order to check whether the assertion after the loop holds, the safety invariant 
must relate the length of the list pointed by $x$  
to the length of the list pointed by $y$. 
Similarly, a termination argument needs to consider the length of the two lists.

The example in Figure~\ref{fig:motivation.d} illustrates how non-linear arithmetic can be encoded via singly-linked lists. 
Thus, the loop in \sm{divides(x, y)} iterates over the list pointed by \sm{x} a number of nodes equal to the 
quotient of the integer division \sm{length(x)/length(y)} such that, after the loop, the list pointed by $z$
has a length equal with the remainder of the division.

The function in Figure~\ref{fig:motivation.e} returns true
iff the list passed in as a parameter is circular.  
The functional correctness of this function is captured by the assertion after the loop checking that 
pointers $p$ and $q$ end up being equal iff the list $l$ is circular.

\section{Theory of Singly Linked Lists with Length}
\label{sec:heap-theory}

In this section we introduce the theory SLH for reasoning about potentially cyclic singly linked lists.

\subsection{Informal Description of SLH}
\label{sec:slh-syntax}

We imagine that there is a set of pointer variables $x, y, \ldots$ which point
to heap cells.  The cells in the heap are arranged into singly linked lists, i.e.
each cell has a ``next'' pointer which points somewhere in the heap.  The lists can be
cyclic and two lists can share a tail, so for example the following heap is allowed in
our logic:

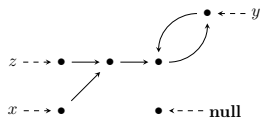
\parbox{\textwidth}{
\centering
\resizebox{.3\textwidth}{!}{
\begin{tikzpicture}

  \node (n1) {$\bullet$};
  \node (n2) [below of=n1] {$\bullet$};
  \node (n3) [right of=n1] {$\bullet$};
  \node (n4) [right of=n3] {$\bullet$};
  \node (n5) [right of=n4,above of=n4] {$\bullet$};
  \node (n6) [below of=n4] {$\bullet$};
  \node (x) [left of=n2] {$x$};
  \node (y) [right of=n5] {$y$};
  \node (z) [left  of=n1] {$z$};
  \node (null) [right=2em of n6] {\bf null};

  \path[draw,->] (n1) edge (n3)
		 (n2) edge (n3)
		 (n3) edge (n4)
		 (n4) edge[bend right=45] (n5)
		 (n5) edge[bend right=45] (n4);

  \path[dashed,->] (x) edge (n2)%
		   (y) edge (n5)%
		   (z) edge (n1)%
		   (null) edge (n6);
 \end{tikzpicture}
}
}
\vspace{1em}

Our logic contains functions for examining the state of the heap, along with the four
standard operations for mutating linked lists: \halloc, \hassign, \hlookup~and \hupdate.
We capture the side-effects of these mutation operators by explicitly naming the
current heap -- we introduce heap variables $h, h'$ etc. which denote the heap in which
each function is to be interpreted.  The mutation operators then become pure functions
mapping heaps to heaps. 
The heap functions of the logic are illustrated by example in Figure~\ref{fig:slh-examples}
and have the following meanings:

\begin{tabular}{rp{.75\textwidth}}
 $\halias(h, x, y)$: & do $x$ and $y$ point to the same cell in heap $h$? \\
 $\hispath(h, x, y)$: & is there a path from $x$ to $y$ in $h$? \\
 $\hpathlen(h, x, y)$: & the length of the shortest path from $x$ to  $y$ in $h$. \\
 $\hisnull(h, x)$: & is $x$ {\bf null} in $h$? \\
 $\hcircular(h, x)$: & is $x$ part of a cycle, i.e. is there some non-empty path from $x$ back to $x$ in $h$? \\
 $h' = \halloc(h, x)$: & obtain $h'$ from $h$ by allocating a new heap cell and reassigning $x$ so that it points to this cell. The newly allocated
 cell is not reachable from any other cell and its successor is {\bf null}. This models the program statement $x = new()$. 
For simplicity, we opt for this allocation policy, but we are not restricted to it. \\
 $h' = \hassign(h, x, y)$: & obtain $h'$ from $h$ by assigning $x$ so that it points to the same cell as $y$.  Models the statement $x = y$. \\
 $h' = \hlookup(h, x, y)$: & obtain $h'$ from $h$ by assigning $x$ to point to $y$'s successor.  Models the statement $x = y{\rightarrow}next$. \\
 $h' = \hupdate(h, x, y)$: & obtain $h'$ from $h$ by updating $x$'s successor to point to $y$.  Models $x{\rightarrow}next = y$.
\end{tabular}

\begin{figure}
\centering

\begin{subfigure}{.45\textwidth}
\begin{framed}
\parbox[t][8em][c]{\textwidth}{
\centering
\vfill
\resizebox{\textwidth}{!}{
 \begin{tikzpicture}
  \node (n1) {$\bullet$};
  \node [right of=n1] (n2) {$\bullet$};
  \node (n3) [right of=n2] {$\bullet$};

  \node (y) [below=1em of n1] {$y$};
  \node (x) [right of=y] {$x$};
  \node (null) [right of=x] {\bf null};

  \path[draw, ->] (n1) edge (n2) (n2) edge (n3);
  \path[dashed,->] (x) edge (n3) (y) edge (n1) (null) edge (n3);

  \node (mid) [below right = .5em and 1em of n3] {$\Rightarrow$};

  \node [above right = .5em and 1em of mid] (m1) {$\bullet$};
  \node (m2) [right of = m1] {$\bullet$};
  \node (m3) [right of = m2] {$\bullet$};

  \node (y) [below=1em of m1] {$y$};
  \node (x) [right of=y] {$x$};
  \node (null) [right of=x] {\bf null};

  \path[draw, ->] (m1) edge (m2) (m2) edge (m3);
  \path[dashed,->] (x) edge (m2) (y) edge (m1) (null) edge (m3);
 \end{tikzpicture}
 }

 \vfill
 $\hlookup(h, x, y)$ \\
 $x = y{\rightarrow}next;$
 }
 \end{framed}
\end{subfigure}
\begin{subfigure}{.45\textwidth}
\begin{framed}
\parbox[t][8em][c]{\textwidth}{
\centering
\vfill
\resizebox{\textwidth}{!}{
 \begin{tikzpicture}[shorten >=1pt,auto]
  \node (n1) {$\bullet$};
  \node (n2) [right of=n1] {$\bullet$};
  \node (n3) [right of=n2] {$\bullet$};

  \node (x) [below of=n1] {$x$};
  \node (null) [below of=n3] {\bf null};

  \path[draw, ->] (n1) edge (n2)
                  (n2) edge (n3);
  \draw[dashed,->] (x) to (n1);
  \draw[dashed,->] (null) to (n3);

  \node [below right of=n3] (mid) {$\Rightarrow$};

  \node (m1) [above right of=mid]{$\bullet$};
  \node (m2) [right of=m1] {$\bullet$};
  \node (m3) [right of=m2] {$\bullet$};
  \node (m4) [below of=m2] {$\bullet$};

  \node (x2) [below =1em of m4] {$x$};
  \node (null) [right of=x2] {\bf null};
each of these 
  \path[draw, ->] (m1) edge (m2) (m2) edge (m3) (m4) edge (m3);
  \draw[dashed,->] (x2) to (m4);
  \draw[dashed,->] (null) to (m3);
 \end{tikzpicture}
}

\vfill
 $\halloc(x)$ \\
 $x = new();$
}
\end{framed}
\end{subfigure}

\begin{subfigure}{.45\textwidth}
\begin{framed}
\centering
\resizebox{\textwidth}{!}{
 \begin{tikzpicture}[shorten >=1pt,auto]
  \node (n1) {$\bullet$};
  \node (n2) [right of=n1] {$\bullet$};
  \node (n3) [right of=n2] {$\bullet$};
  \node (n4) [below=2em of n2] {$\bullet$};

  \node (y) [below=1em of n4] {$y$};
  \node (x) [left of=y] {$x$};
  \node (null) [right of=y] {\bf null};

  \path[draw, ->] (n1) edge (n2)
                  (n2) edge (n3)
                  (n4) edge (n3);
  \path[dashed,->] (x) edge (n1)
                   (y) edge (n4)
                   (null) edge (n3);

  \node [below right of=n3] (mid) {$\Rightarrow$};

  \node (m1) [above right of=mid]{$\bullet$};
  \node (m2) [right of=m1] {$\bullet$};
  \node (m3) [right of=m2] {$\bullet$};
  \node (m4) [below=2em of m2] {$\bullet$};

  \node (y2) [below=1em of m4] {$y$};
  \node (x2) [left of=y2] {$x$};
  \node (null) [right of=y2] {\bf null};

  \path[draw, ->] (m1) edge (m2)
                  (m2) edge (m3)
                  (m4) edge (m3);
  \path[dashed,->] (x2) edge (m4)
                   (y2) edge (m4)
                   (null) edge (m3);
 \end{tikzpicture}
}

  $\hassign(x, y)$\\
  $x = y;$
 \end{framed}
\end{subfigure}
\begin{subfigure}{.45\textwidth}
\begin{framed}
\centering
\resizebox{\textwidth}{!}{
 \begin{tikzpicture}[shorten >=1pt,auto]
  \node (n1) {$\bullet$};
  \node (n2) [right of=n1] {$\bullet$};
  \node (n3) [right of=n2] {$\bullet$};
  \node (n4) [below=2em of n2] {$\bullet$};

  \node (y) [below=1em of n4] {$y$};
  \node (x) [left of=y] {$x$};
  \node (null) [right of=y] {\bf null};

  \path[draw, ->] (n1) edge (n2)
                  (n2) edge (n3)
                  (n4) edge (n3);
  \path[dashed,->] (x) edge (n1)
                   (y) edge (n4)
                   (null) edge (n3);

  \node [below right of=n3] (mid) {$\Rightarrow$};

  \node (m1) [above right of=mid]{$\bullet$};
  \node (m2) [right of=m1] {$\bullet$};
  \node (m3) [right of=m2] {$\bullet$};
  \node (m4) [below=2em of m2] {$\bullet$};

  \node (y2) [below=1em of m4] {$y$};
  \node (x2) [left of=y2] {$x$};
  \node (null) [right of=y2] {\bf null};

  \path[draw, ->] (m1) edge (m4)
                  (m2) edge (m3)
                  (m4) edge (m3);
  \path[dashed,->] (x2) edge (m1)
                   (y2) edge (m4)
                   (null) edge (m3);
 \end{tikzpicture}
}

  $\hupdate(x, y)$\\
  $x{\rightarrow}next = y;$
\end{framed}
\end{subfigure}

\begin{subfigure}{.91\textwidth}
\begin{framed}
\centering
\begin{minipage}[c]{.4\textwidth}
\begin{tikzpicture}[>=stealth',shorten >=1pt,auto,scale=.8]

  \node (n1) {$\bullet$};
  \node (n2) [below of=n1] {$\bullet$};
  \node (n3) [right of=n1] {$\bullet$};
  \node (n4) [right of=n3] {$\bullet$};
  \node (n5) [right of=n4,above of=n4] {$\bullet$};
  \node (n6) [below of=n4] {$\bullet$};
  \node (x) [below of=n2] {$x$};
  \node (y) [right of=n5] {$y$};
  \node (z) [above of=n1] {$z$};
  \node (null) [below of=n6] {\bf null};

  \path[draw,->] (n1) edge (n3)
		 (n2) edge (n3)
		 (n3) edge (n4)
		 (n4) edge[bend right=45] (n5)
		 (n5) edge[bend right=45] (n4);

  \path[dashed,->] (x) edge (n2)%
		   (y) edge (n5)%
		   (z) edge (n1)%
		   (null) edge (n6);
 \end{tikzpicture}
 \end{minipage}
 \begin{minipage}[c]{.4\textwidth}
 \vfill
 \begin{align*}
 \hpathlen(x, y) & =  3 \\
 \hispath(z, y) & = \text{true} \\
 \hispath(x, z) & = \text{false} \\
 \halias(x, z) & = \text{false} \\
 \hisnull(x) & = \text{false} \\
 \hcircular(y) & = \text{true}
 \end{align*}
 \vfill
 \end{minipage}
\end{framed}
\end{subfigure}

\caption{SLH by example\label{fig:slh-examples}}
\end{figure}
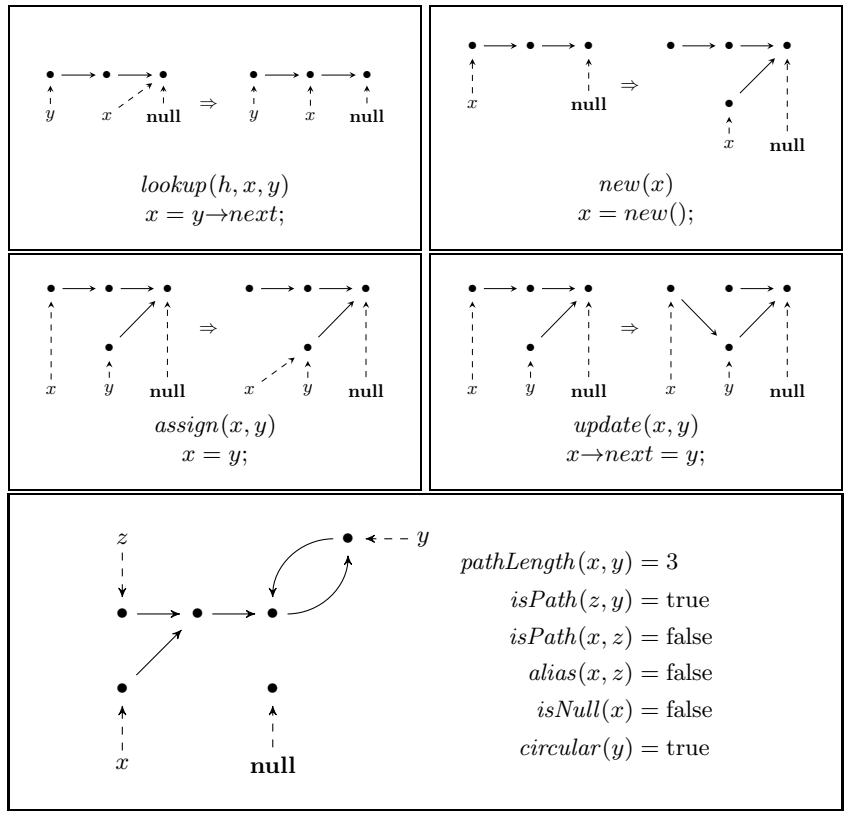

%
%
%
%
%

\subsection{Syntax of SLH}


The theory of singly-linked lists with length, SLH, uses a background arithmetic theory $\mathcal{T_{B}}$ for the length of lists (implicitly any numeric variable). 
Thus, SLH has the following signature:
\[
\begin{array}{ll}
\Sigma_{SLH} =  \Sigma_{B} \cup & \{\halias(\cdot,\cdot,\cdot), \hispath(\cdot,\cdot,\cdot), \hpathlen(\cdot,\cdot,\cdot), \hisnull(\cdot, \cdot), \\
& \hcircular(\cdot, \cdot), \cdot{=}\halloc(\cdot,\cdot), \cdot{=}\hassign(\cdot,\cdot, \cdot), \\
& \cdot{=}\hlookup(\cdot,\cdot, \cdot), \cdot{=}\hupdate(\cdot,\cdot, \cdot)\}.  
\end{array}
\]
where the nine new symbols correspond to the heap-specific functions
described in the previous section.

\paragraph{Sorts.} Heap variables (e.g. $h$ in $\halias(h,x,y)$) have sort $\mathcal{S_{H}}$, 
pointer variables have sort $\mathcal{S}_{Addr}$ (e.g. $x$ and $y$ in $\halias(h,x,y)$), 
numeric variables have sort $\mathcal{S_B}$ (e.g. $n$ in $n = \hpathlen(h,x,y)$).

\paragraph{Literal and formula.} A literal in SLH is either a heap function or a $\mathcal{T_{B}}$-literal.
A formula in SLH is a boolean combination of SLH-literals.

\subsection{Semantics of SLH}
\label{sec:semantics}

We give the semantics of SLH by defining the models in which an SLH formula holds.
An interpretation $\Gamma$ is a function mapping free variables to elements of the appropriate sort.
If an SLH formula $\phi$ holds in some interpretation $\Gamma$,
we say that $\Gamma$ \emph{models} $\phi$ and write $\Gamma \models \phi$. 

Interpretations may be constructed using the following substitution rule:

\[
 \Gamma[h \mapsto H](x) = \begin{cases}
                        H & \text{if } x = h \\
                        \Gamma(x) & \text{otherwise}
                       \end{cases}
\]

Pointer variables are considered to be a set of constant symbols
and are thus given a fixed interpretation. 
The only thing that matters is that their interpretation is pairwise different.
We assume that the pointer variables include a special name {\bf null}.  The set of pointer variables is denoted by
the symbol $P$.

We will consider the semantics of propositional logic to be standard and the
semantics of $\mathcal{T_{B}}$ given, and thus just define the semantics of
heap functions.  To do this, we will first define the class of objects that
will be used to interpret heap variables.

\begin{definition}[Heap]
A heap over pointer variables $P$ is a pair $H = \langle L, G \rangle$.
$G$ is a finite graph with vertices $V(G)$ edges $E(G)$.
$L : P \to V(G)$ is a labelling function mapping each pointer variable to a vertex of $G$.
We define the cardinality of a heap to be the cardinality of its underlying graph's
vertices:  $\|H\| = \|V(G)\|$.
\end{definition}

\begin{definition}[Singly Linked Heaps]
 A heap $H = \langle L, G \rangle$ is a singly linked heap iff
 each vertex has outdegree 1, except for a single sink vertex
 that has outdegree 0 and is labelled by {\bf null}:
 \begin{align*}
  \forall v \in V(G) . & (\mathrm{outdegree}(v) = 1 \wedge L(\mathbf{null}) \neq v) \vee \\
                       & (\mathrm{outdegree}(v) = 0 \wedge L(\mathbf{null}) = v)
 \end{align*}
\end{definition}

%

Having defined our domain of discourse, we are now in a position to define
the semantics of the various heap functions introduced in
Section~\ref{sec:slh-syntax}.  We begin with the functions examining the
state of the heap and will use a standard structural recursion to give the
semantics of the functions with respect to an implicit interpretation
$\Gamma$, so that $\sem{h} = \Gamma(h)$.  We will use the shorthand $u
\stackrel{n}{\rightarrow} v$ to say that if we start at node $u$, then
follow $n$ edges, we arrive at $v$.  We also use $L(H)$ to select the
labelling function $L$ from $H$:
\begin{align*}
u \stackrel{n}{\rightarrow} v & \defeq \langle u, v \rangle \in E^n \\
u \rightarrow^* v & \defeq \exists n \geq 0 . u \stackrel{n}{\rightarrow} v \\
u \rightarrow^+ v & \defeq \exists n > 0 . u \stackrel{n}{\rightarrow} v
\end{align*}

Note that $u \stackrel{0}{\rightarrow} u$.  The semantics of the heap functions are then:
\begin{align*}
 \sem{\hpathlen(h, x, y)}\Gamma & \defeq \min \left( \{ n \mid L(\sem{h})(x) \stackrel{n}{\rightarrow} L(\sem{h})(y) \} \cup \{ \infty \} \right) \\
 \sem{\hcircular(h, x)}\Gamma & \defeq \exists v \in V(\sem{h}) . L(\sem{h})(x) \rightarrow^+ v \wedge v \rightarrow^+ L(\sem{h})(x) \\
 \sem{\halias(h, x, y)}\Gamma & \defeq \sem{\hpathlen(h, x, y)} = 0 \\
 \sem{\hispath(h, x, y)}\Gamma & \defeq \sem{\hpathlen(h, x, y)} \neq \infty \\
 \sem{\hisnull(h, x)}\Gamma & \defeq \sem{\hpathlen(h, x, \hnull)} = 0
\end{align*}

Note that since the graph underlying $H$ has outdegree 1, \hpathlen~and
\hcircular~can be computed in $O(\| H \|)$ time, or equivalently they can be
encoded with $O(\| H \|)$ arithmetic constraints.

To define the semantics of the mutation operations, we will consider
separately the effect of each mutation on each component of the heap -- the
labelling function $L$, the vertex set $V$ and the edge set $E$.  Where a
mutation's effect on some heap component is not explicitly stated, the
effect is \id.  For example, \hassign~does not modify the vertex set, and so
$\hassign_V = \id$.  In the following definitions, we will say that
$\mathrm{succ}(v)$ is the unique vertex such that $(v, \mathrm{succ}(v)) \in
E(H)$.

\parbox{\textwidth}{
\centering
\small
\begin{align*}
 \sem{\halloc_V(h, x)}\Gamma & \defeq V(\sem{h}) \cup \{ q \} \text{where $q$ is a fresh vertex} \\
 \sem{\halloc_E(h, x)}\Gamma & \defeq E(\sem{h}) \cup \{ (q, \hnull) \} \\
 \sem{\halloc_L(h, x)}\Gamma & \defeq L(\sem{h})[x \mapsto q] \\
 \sem{\hassign_L(h, x, y)}\Gamma & \defeq L(\sem{h})[x \mapsto L(\sem{h})(y)] \\
 \sem{\hlookup_L(h, x, y)}\Gamma & \defeq L(\sem{h})[x \mapsto \mathrm{succ}(y)] \\
 \sem{\hupdate_E(h, x, y)}\Gamma & \defeq \left( E(\sem{h}) \setminus \{ (L(\sem{h})(x), \mathrm{succ}(L(\sem{h})(x))) \} \right) \cup \\
&  \{ (L(\sem{h})(x), L(\sem{h})(y)) \}
\end{align*}
}


%
%
%
%
%

\section{Deciding Validity of SLH}

We will now turn to the question of deciding the validity of an SLH formula,
that is for some formula $\phi$ we wish to determine whether $\phi$ is a
tautology or if there is some $\Gamma$ such that $\Gamma \models \neg \phi$. 
To do this, we will show that SLH enjoys a finite model property and that
the existence of a fixed-size model can be encoded directly as an arithmetic
constraint.

Our high-level strategy for this proof will be to define progressively
coarser equivalence relations on SLH heaps that respect the transformers and
observation functions.  The idea is that all of the heaps in a particular
equivalence class will be equivalent in terms of the SLH formulae they
satisfy.  We will eventually arrive at an equivalence relation
(homeomorphism) that is sound in the above sense and which is also
guaranteed to have a small heap in each equivalence class.

%

From here on we will slightly generalise the definition of a singly linked
heap and say that the underlying graph is weighted with weight function $w :
E(H) \to \mathbb{N}$.  When we omit the weight of an edge (as we have in all
heaps up until now), it is to be understood that the edge's weight is 1.

\subsection{Sound Equivalence Relations}
We will say that an equivalence relation $\approx$ is \emph{sound} if the following conditions hold
for each pair of pointer variables $x, y$ and transformer $\tau$:
\begin{align}
 \forall H, H' \cdot  H \approx H' \Rightarrow & \hpathlen(H, x, y) = \hpathlen(H', x, y) ~ \wedge \\
               & \hcircular(H, x) = \hcircular(H', x) ~ \wedge \\
               & \tau(H) \approx \tau(H')
\end{align}

The first two conditions say that if two heaps are in the same equivalence
class, there is no observation that can distinguish them.  The third
condition says that the equivalence relation is inductive with respect to
the transformers.  There is therefore no sequence of transformers and
observations that can distinguish two heaps in the same equivalence class.

We begin defining two sound equivalence relations:

\begin{definition}[Reachable Sub-Heap]
 The reachable sub-heap $H|_P$ of a heap $H$ is $H$ with vertices restricted to those
 reachable from the pointer variables $P$:
 $$V(H|_P) = \{ v  \mid \exists p \in P . \langle L(p), v \rangle \in E^* \}$$
\end{definition}

Then the relation $\{ \langle H, H' \rangle \mid H|_P = H'|_P \}$ is sound.

\begin{definition}[Heap Isomorphism]
 Two heaps $H = \langle L, G \rangle, H' = \langle G', L' \rangle$ are
 isomorphic (written $H \simeq H'$) iff there exists a graph isomorphism $f
 : G|_P \to G'|_P$ that respects the labelling function, i.e., $\forall p
 \in P .  f(L(p)) = L'(p)$.
\end{definition}

\begin{example}
$H$ and $H'$ are not isomporphic, even though their underlying graphs are.

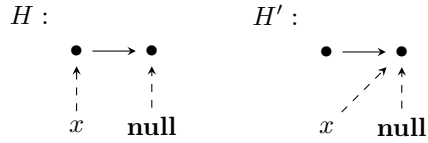
\parbox{.9\textwidth}{
\centering
 \begin{tikzpicture}
  \node (h) {$H:$};

  \node (n1) [below right =.1em and .1em of h] {$\bullet$};
  \node (n2) [right of=n1] {$\bullet$};

  \node (x) [below of=n1] {$x$};
  \node (null) [below of=n2] {\bf null};

  \path[draw, ->] (n1) edge (n2);
  \path[dashed,->] (x) edge (n1) 
		   (null) edge (n2);

  \node (h2) [right=7em of h] {$H':$};

  \node (m1)[below right =.1em and .1em of h2] {$\bullet$};
  \node (m2) [right of=m1] {$\bullet$};

  \node (x2) [below of=m1] {$x$};
  \node (null2) [below of=m2] {\bf null};

  \path[draw, ->] (m1) edge (m2);
  \path[dashed,->] (x2) edge (m2)
		   (null2) edge (m2);
  \end{tikzpicture}
}
\end{example}

\begin{theorem}
\label{thm:iso-sound}
Heap isomorphism is a sound equivalence relation.
\end{theorem}

\subsection{Heap Homeomorphism}

The final notion of equivalence we will describe is the weakest.  Loosely,
we would like to say that two heaps are equivalent if they are ``the same
shape'' and if the shortest distance between pointer variables is the same. 
To formalise this relationship, we will be using an analog of topological
homeomorphism.

\begin{definition}[Edge Subdivision]
 A graph $G'$ is a subdivison of $G$ iff $G'$ can be obtained by repeatedly subdividing edges in $G$,
 i.e., for some edge $(u, v) \in E(G)$ introducing a fresh vertex $q$ and
 replacing the edge $(u, v)$ with edges $(u, q), (q, v)$ such that $w'(u, q) + w'(q, v) = w(u, v)$.
 Subdivision for heaps is defined in terms of their underlying graphs.
\end{definition}

We define a function \hsubdivide, which subdivides an edge in a heap.  As usual,
the function is defined componentwise on the heap:
\begin{align*}
 \hsubdivide_V(H, u, v, k) = & V \cup \{ q \} \\
 \hsubdivide_E(H, u, v, k) = & \left( E \setminus \{ (u, v) \} \right) \cup \{ (u, q), (q, v) \} \\
 \hsubdivide_W(H, u, v, k) = & W(H)[(u, v) \mapsto \infty, (u, q) \mapsto k, (q, v) \mapsto W(H)(u, v) - k]
\end{align*}

The dual of edge subdivision is edge \emph{smoothing} -- if we have two
edges \mbox{$u \stackrel{n}{\longrightarrow} q \stackrel{m}{\longrightarrow}
v$}, where $q$ is unlabelled and has no other incoming edges, we can remove
$q$ and add the single edge $u \stackrel{n+m}{\longrightarrow} v$.

\begin{example}
$H'$ is a subdivision of $H$.

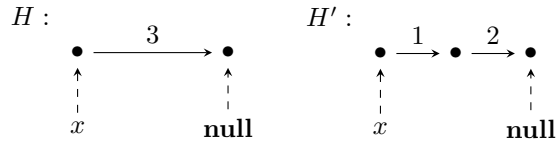
\parbox{.9\textwidth}{
\centering

 \begin{tikzpicture}
  \node (h) {$H:$};

  \node (n1) [below right =.1em and .1em of h] {$\bullet$};
  \node (space) [right of=n1] {};
  \node (n2) [right of=space] {$\bullet$};

  \node (x) [below of=n1] {$x$};
  \node (null) [below of=n2] {\bf null};

  \path[draw, ->] (n1) to node[above] {3} (n2);
  \path[dashed,->] (x) edge (n1)
                   (null) edge (n2);

  \node (h2) [right=9em of h] {$H':$};

  \node (m1)[below right =.1em and .1em of h2] {$\bullet$};
  \node (m2) [right of=m1] {$\bullet$};
  \node (m3) [right of=m2] {$\bullet$};

  \node (x2) [below of=m1] {$x$};
  \node (null2) [below of=m3] {\bf null};

  \path[draw, ->] (m1) edge node[above] {1} (m2)
                  (m2) edge node [above] {2} (m3);
  \path[dashed,->] (x2) edge (m1)
		   (null2) edge (m3);
 \end{tikzpicture}
}
\end{example}

\begin{lemma}[Subdividing an Edge Preserves Observations]
\label{lem:subdivide-pathlen}
 If $H'$ is obtained from $H$ by subdividing one edge, then for any $x, y$ we have:
 \begin{align}
  \hpathlen(H, x, y) & = \hpathlen(H', x, y) \\
  \hcircular(H, x) & = \hcircular(H', x)
 \end{align}
\end{lemma}

\begin{definition}[Heap Homeomorphism]
 Two heaps $H, H'$ are homeomorphic (written $H \sim H'$) iff there there is a heap isomorphism from some subdivision of $H$ to some subdivision of $H'$.
\end{definition}

Intuitively, homeomorphisms preserve the topology of heaps: if two heaps are homeomorphic
then they have the same number of loops and the same number of ``joins'' (vertices with indegree $\geq 2$).

\begin{example}
\label{ex:homeomorphic}
 $H$ and $H'$ are homeomorphic, since they can each be subdivided to produce $S$.

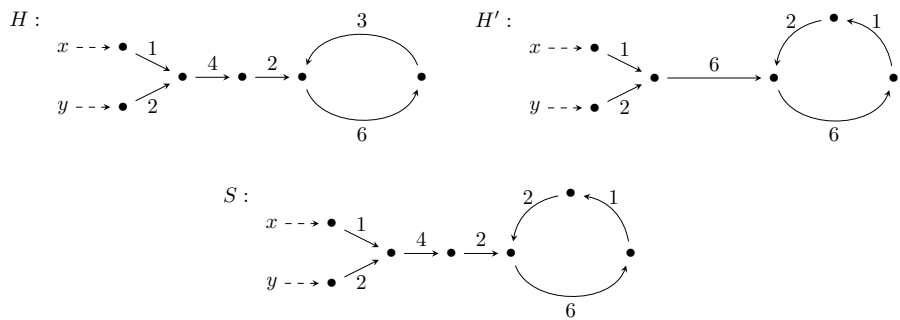
 \parbox{.9\textwidth}{
 \centering

 \resizebox{\textwidth}{!}{
 \begin{tikzpicture}
  \node (h) {$H:$};

  \node (x) [below right =.1em and .1em of h] {$x$};
  \node (y) [below of=x] {$y$};
  
  \node (n1) [right of=x] {$\bullet$};
  \node (n2) [below of=n1] {$\bullet$};

  \path (n1) to node (mid) {} (n2);

  \node (n3) [right of=mid] {$\bullet$};
  \node (n4) [right of=n3] {$\bullet$};

  \node (n5) [right of=n4] {$\bullet$};
  \node (space2) [right of=n5] {};
  \node (n6) [right of=space2] {$\bullet$};

  \path[->,draw] (n1) edge node [above] {1} (n3)
                 (n2) edge node [below] {2} (n3)
                 (n3) edge node [above] {4} (n4)
                 (n4) edge node [above] {2} (n5)
                 (n5) edge[bend right=70] node [below] {6} (n6)
                 (n6) edge[bend right=70] node [above] {3} (n5);

  \path[->,dashed] (x) edge (n1)
                   (y) edge (n2);

  \node (h2) [right=20em of h] {$H':$};

  \node (x2) [below right =.1em and .1em of h2] {$x$};
  \node (y2) [below of=x2] {$y$};

  \node (m1) [right of=x2] {$\bullet$};
  \node (m2) [below of=m1] {$\bullet$};

  \path (m1) to node (mid2) {} (m2);

  \node (m3) [right of=mid2] {$\bullet$};
  \node (m4) [right of=m3] {};

  \node (m5) [right of=m4] {$\bullet$};
  \node (space2) [right of=m5] {};
  \node (m6) [right of=space2] {$\bullet$};

  \node (m7) [above of=space2] {$\bullet$};

  \path[->,draw] (m1) edge node [above] {1} (m3)
                 (m2) edge node [below] {2} (m3)
                 (m3) edge node [above] {6} (m5)
                 (m5) edge[bend right=70] node [below] {6} (m6)
                 (m6) edge[bend right=35] node [above] {1} (m7)
                 (m7) edge[bend right=35] node [above] {2} (m5);

  \path[->,dashed] (x2) edge (m1)
                   (y2) edge (m2);

  \node (s) [below right=7em and 8em of h] {$S:$};

  \node (x) [below right =.1em and .1em of s] {$x$};
  \node (y) [below of=x] {$y$};

  \node (n1) [right of=x] {$\bullet$};
  \node (n2) [below of=n1] {$\bullet$};

  \path (n1) to node (mid) {} (n2);

  \node (n3) [right of=mid] {$\bullet$};
  \node (n4) [right of=n3] {$\bullet$};

  \node (n5) [right of=n4] {$\bullet$};
  \node (space2) [right of=n5] {};
  \node (n6) [right of=space2] {$\bullet$};

  \node (n7) [above of=space2] {$\bullet$};

  \path[->,draw] (n1) edge node [above] {1} (n3)
                 (n2) edge node [below] {2} (n3)
                 (n3) edge node [above] {4} (n4)
                 (n4) edge node [above] {2} (n5)
                 (n5) edge[bend right=70] node [below] {6} (n6)
                 (n6) edge[bend right=35] node [above] {1} (n7)
                 (n7) edge[bend right=35] node [above] {2} (n5);

  \path[->,dashed] (x) edge (n1)
                   (y) edge (n2);
%
%
%
%
%
%
 \end{tikzpicture}
 }
 }
\end{example}

\begin{lemma}[Transformers Respect Homeomorphism]
\label{lem:homeo-transform}
For any heap transformer $\tau$, if $H_1 \sim H_2$ then $\tau(H_1) \sim \tau(H_2)$.
\end{lemma}

\begin{proof}
It suffices to show that for any transformer $\tau$ and single-edge subdivision $s$, the
following diagram commutes:

\parbox{.9\textwidth}{
\centering

\begin{tikzpicture}[node distance=5em]
 \node (A) {$A$};
 \node (B) [below of=A] {$B$};
 \node (C) [right of=A] {$C$};
 \node (D) [below of=C] {$D$};

 \path[->,draw] (A) edge node[above] {$\tau$} (C)
                 (A) edge node[left] {$s$} (B)
                 (B) edge node[below] {$\tau$} (D)
                 (C) edge node[right] {$s$} (D);
\end{tikzpicture}
}

We will check that $\tau \circ s = s \circ \tau$ by considering the components of each arrow
separately and using the semantics defined in Section~\ref{sec:semantics}.  The only
difficult case is for \hlookup, for which we provide the proof in full.  This case is
illustrative of the style of reasoning used for the proofs of the 
other transformers, which can be found in the extended version of this paper.


%
%
%

\paragraph{$\tau = \hlookup$:}

Now that we have weighted heaps, there are two cases for \mbox{\hlookup:} if the edge
leaving $L(y)$ does not have weight 1, we need to first subdivide so that it does; otherwise
the transformer is exactly as in the unweighted case, which can be seen easily to commute.

In the second (unweighted) case, all of the components commute due to $\id$.  Otherwise,
\hlookup is a composition of some subdivision $s'$ and then unweighted lookup:
$\hlookup = \hlookup_U \circ s'$.

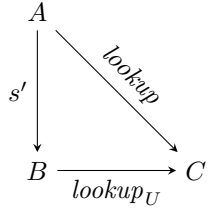
\parbox{.9\textwidth}{
\centering

\begin{tikzpicture}[node distance=6em]
 \node (A) {$A$};
 \node (B) [below of=A] {$B$};
 \node (C) [right of=B] {$C$};

 \path[->,draw] (A) edge node[left] {$s'$} (B)
                    edge node[above,sloped] {$\hlookup$} (C)
                (B) edge node [below] {$\hlookup_U$} (C);
\end{tikzpicture}
}

\vspace{1em}

Our commutativity condition is then:
\begin{align*}
(\hlookup_U \circ s') \circ s = & s \circ (\hlookup_U \circ s') \\
\intertext{We know that unweighted \hlookup commutes with arbitrary subdivisions, so}
(\hlookup_U \circ s') \circ s = & s \circ (s' \circ \hlookup_U) \\
\hlookup_U \circ (s' \circ s) = & (s \circ s') \circ \hlookup_U \\
\intertext{But the composition of two subdivisions is a subdivision, so we are done.}
\end{align*}
\end{proof}

\begin{theorem}
 \label{thm:homeo-sound}
 Homeomorphism is a sound equivalence relation.
\end{theorem}

\begin{proof}
 This is a direct consequence of Lemma~\ref{lem:subdivide-pathlen} and Lemma~\ref{lem:homeo-transform}.
\end{proof}
%
%

We would now like to show that for each equivalence class induced by $\sim$,
there is a unique minimal element.  This is easy to show if we consider the
category {\bf SLH} of singly linked heaps, with edge subdivisions as arrows. 
To recap the definition of a category:

\begin{itemize}
 \item A \emph{category} is a collection of objects along with a collection of arrows, which are maps from
 one object to another.  We can compose any two arrows $A \rightarrow B \rightarrow C$ to form one arrow $A \rightarrow C$.
 Composition is associative and each object has an arrow $A \stackrel{\mathrm{id}}{\rightarrow} A$.
 \item A diagram is a picture consisting of some objects along with some arrows between the objects.  We say that
 the diagram \emph{commutes} if for any two paths between two objects, the arrow generated by composing all the
 arrows on the first path equals the composition of all the arrows on the second path.
\end{itemize}

As mentioned, objects of the category {\bf SLH} are singly linked heaps, and there is an arrow from one heap to another
if the first can be subdivided into the second.  To illustrate, Example~\ref{ex:homeomorphic}
is represented in {\bf SLH} by the following diagram:

\parbox{.9\textwidth}{
 \centering

 \begin{tikzpicture}[node distance=5em]
  \node (h)  {$H$};
  \node (x) [below of=h] {$S$};
  \node (h') [left of=x] {$H'$};

  \path[->,draw] (h) edge (x)
                 (h') edge (x);
 \end{tikzpicture}
}

Now for every pair of homeomorphic heaps $H_1 \sim H_2$ we know that there
is some $X$ that is a subdivision of both $H_1$ and $H_2$.  Clearly if we
continue subdividing edges, we will eventually arrive at a heap where every
edge has weight 1, at which point we will be unable to subdivide any
further.  Let us call this maximally subdivided heap the \emph{shell}, which
we will denote by $\mathrm{Sh}(H_1)$.  Then $\mathrm{Sh}(H_1) =
\mathrm{Sh}(H_2)$ is the pushout of the previous diagram.  Dually, there is
some $Y$ that both $H_1$ and $H_2$ are subdivisions of, and the previous
diagram has a pullback, which we shall call the \emph{kernel}.  This is the
heap in which all edges have been smoothed.  The following diagram commutes,
and since a composition of subdivisions and smoothings is a homeomorphism, 
all of the arrows (and their inverses) in this diagram are homeomorphisms.
In fact, the $H_1, H_2, X, Y, \mathrm{Sh}$ and $\mathrm{Ke}$ are exactly
an equivalance class:


\parbox{.9\textwidth}{
 \centering

 \begin{tikzpicture}[node distance=5em]
  \node (x) {$Y$};
  \node (h1) [right of=x] {$H_1$};
  \node (h2) [below of=x] {$H_2$};

  \node (y) [right of=h2] {$X$};

  \node (ke) [above left=1em and 1em of x] {Ke};
  \node (sh) [below right=1em and 1em of y] {Sh};

  \path[->] (x) edge (h1) edge (h2)
            (h1) edge (y)
            (h2) edge (y)
            (ke) edge (x) edge[bend left] (h1) edge[bend right] (h2)
            (h1) edge[bend left] (sh)
            (h2) edge[bend right] (sh)
            (y) edge (sh);

  \path[<->,dashed] (h1) edge node[above] {$\sim$} (h2);
 \end{tikzpicture}
}


%
%

\begin{lemma}[Kernels are Small]
\label{lem:small-kernels}
 For any $H$, $\|\mathrm{Ke}(H)\| \leq 2 \times \|P\|$.
\end{lemma}

\begin{proof}
 Since $\mathrm{Ke}(H)$ is maximally smoothed, every unlabelled vertex has indegree $\geq 2$.
 We will partition the vertices of $H$ into named and unlabelled vertices:
 \begin{align*}
 N = & \{ v \in V(H) \mid \exists p \in P . L(p) = v \} \\
 U = & \{ u \in V(H) \mid \forall p \in P . L(p) \neq u \} \\
 V(H) = & N \cup U
 \end{align*}

 Then let $n = \|N\|$ and $u = \|U\|$.
 Now, the total indegree of the underlying graph must be equal to the total outdegree, so:
 \begin{align*}
  \sum_{v \in V(H)}\!\! \mathrm{out}(v) & = \sum_{v \in V(H)}\!\! \mathrm{in}(v) \\
  n + u & = \sum_{n \in N} \mathrm{in}(n) +  \sum_{u \in U} \mathrm{in}(u) \\
        & = \sum_{n \in N} \mathrm{in}(n) + 2u + k \\
        & \intertext{where $k \geq 0$, since $\mathrm{in}(u) \geq 2$ for each $u$.} \\
  n     & = u + \underbrace{\sum_{n \in N} \mathrm{in}(n) + k}_{\geq 0} \\
  n     & \geq u
 \end{align*}
 So $u \leq n \leq \| P \|$, hence $\| \mathrm{Ke}(H) \| = n + u \leq 2 \times \|P \|$.
\end{proof}

\begin{theorem}[SLH has Small Model]
\label{thm:small-model}
 For any SLH formula $\forall h . \phi$, if there is a counterexample $\Gamma \models \neg \phi$ then there is
 $\Gamma' \models \neg \phi$ with every heap-sorted variable in $\Gamma$ being interpreted by a homeomorphism kernel.
\end{theorem}

\begin{proof}
 This follows from Theorem~\ref{thm:homeo-sound} and Lemma~\ref{lem:small-kernels}.
\end{proof}

We can encode the existence of a small model with an arithmetic constraint
whose size is linear in the size of the SLH formula, since each of the
transformers can be encoded with a constant sized constraint and the
observation functions can be encoded with a constraint of size $O(\| H \|) =
O(\| P \|)$.  An example implementation of the constraints used to encode
each atom is given in Section~\ref{sec:implementation}.  We need one
constraint for each of the theory atoms, which leaves us with $O(\| P \|
\times \| \phi \|)$ constraints.

\begin{corollary}[Decidability of SLH] \label{thm:decidability}
 If the background theory $\mathcal{T_B}$ is decidable, then SLH is decidable. 
\end{corollary}

\begin{proof}
The existence of a small model can be encoded with a linear number of arithmetic constraints in $\mathcal{T_B}$.
\end{proof}

\section{Using SLH for Verification}
Our intention is to use SLH for reasoning about the safety and termination of programs with 
potentially cyclic singly-linked lists:
\begin{itemize} 
\item For safety, we annotate loops with safety invariants and 
generate VCs checking that each loop annotation is genuinely a safety invariant, i.e.
(1) it is satisfied by each state reachable on entry to the loop, (2) it is inductive with respect to the program's transition relation, 
and (3) excludes any states where an assertion violation takes place (the assertions include those ensuring memory safety). 
The existence of a safety invariant corresponds to the notion of
partial correctness: no assertion fails, but the program may never stop running.
\item For termination, we provide ranking functions for each loop and generate VCs to check
that the loops do terminate, i.e. the ranking function is monotonically decreasing with respect to the loop's body and (2) it is bounded from below. 
By combining these VCs with those generated for safety, we create a total-correctness specification.
\end{itemize}

The two additional items we must provide in order to be able to generate these VCs, are
 a programming language and  the strongest post-condition for formulae in SLH with respect to statements in the programming language. We do so next.  


\subsection{Programming Language} \label{sec:prog_lang}
We use the sequential programming language in Fig.~\ref{fig.language}.  It
allows heap allocation and mutation, with \sm{v} denoting a variable and
\sm{next} a pointer field.  To simplify the presentation,
we assume each data structure has only one pointer field, \sm{next}, and allow only one-level field access, denoted by
\sm{v{\rightarrow}next}.  Chained dereferences of the form
\sm{v{\rightarrow}next{\rightarrow}next{\ldots}} are handled by introducing
auxiliary variables.  The statement \sm{\hassert(\phi)} checks whether 
$\phi$ (expressed in the heap theory described in Section~\ref{sec:heap-theory}) 
holds for the current program state, whereas
\sm{\hassume(\phi)} constrains the program state.

For ease of use, when using SLH in the context of safety 
and termination verification, the SLH functions 
we expose in the specification language are side-effect free.  That is to say, we don't require the explicit 
heap $h$ to be mentioned in the specifications. 

\begin{figure*}[t]
\begin{center}
\fbox{
$
\begin{array}{lll}
datat &:=& struct~ C~ \{(typ~ v)^*\}\\
e &:=& v ~|~ v{\rightarrow}next ~|~ \mathrm{new}(C) ~|~ \nil\\ 
S &:=& v {=} e ~|~  v_1{\rightarrow}next {=} v_2 ~|~ S_1;S_2 ~|~ \hif~ (B)~ S_1~ \helse~ S_2 ~|~\\
&& \hwhile~ (B)~ S ~|~ \mathrm{assert}(\phi) ~|~ \mathrm{assume}(\phi) \\
\end{array}
$
}
\caption{Programming Language}\label{fig.language}
\end{center}
\end{figure*}

\subsection{Strongest Post-Condition}
To create a verification condition from a specification, we first decompose the
specification into Hoare triples and then compute the strongest post-condition
to generate a VC in the SLH theory.  Since SLH includes primitive operations for
heap manipulation, our strongest post-condition is easy to compute:

\begin{align*}
\mathsf{SP}(x = y, \phi) & \defeq h' = \hfresh \wedge \phi[h'/h] \wedge h = \hassign(h', x, y) \\
\mathsf{SP}(x = y{\rightarrow}next, \phi) & \defeq h' = \hfresh \wedge \phi[h'/h] \wedge h = \hlookup(h', x, y) \\
\mathsf{SP}(x = \mathrm{new}(C), \phi) & \defeq h' = \hfresh \wedge \phi[h'/h] \wedge h = \hnew(h', x, y) \\
\mathsf{SP}(x{\rightarrow}next = y, \phi) & \defeq h' = \hfresh \wedge \phi[h'/h] \wedge h = \hupdate(h', x, y) \\
\end{align*}

Where \hfresh introduces a fresh heap variable.  The remaining cases for {\sf SP} are standard.

\subsection{VC Generation Example}

\begin{figure}
\centering
\begin{tabular}{c}
\begin{lstlisting}[mathescape=true]
x = y;

while (x $\neq$ null) {
  $\{ \hispath(y, x) \}$
  x = x$\rightarrow$next;
}

assert($\hispath(y, x)$);
\end{lstlisting}
\end{tabular}
\caption{An annotated program.\label{fig:annotated}}
\end{figure}

Consider the program in Figure~\ref{fig:annotated} which has been annotated with a loop invariant.
In order to verify the partial-correctness condition that the assertion cannot fail, we must check
the following Hoare triples:

\begin{align}
 \{ \top \} \: \lstinline|x = y| \: \{ \hispath(y, x) \} & \\
 \{ \hispath(y, x) \wedge \neg \hisnull(x)\} \:  \lstinline|x = x|\rightarrow\lstinline|next| \: \{ \hispath(y, x) \} & \\
 \{ \hispath(y, x) \wedge \hisnull(x) \}  \: \mathbf{skip} \: \{ \hispath(y, x) \} &
\end{align}

Taking strongest post-condition across each of these triples generates the following SLH VCs:

\begin{align}
 \forall h . h' = \hassign(h, x, y) & \Rightarrow \hispath(h', y, x) \\
 \forall h . \hispath(h, y, x) \wedge \neg \hisnull(x) \wedge h' = \hlookup(h, x, x) & \Rightarrow \hispath(h', y, x) \\
 \forall h . \hispath(h, y, x) \wedge \hisnull(x) & \Rightarrow \hispath(h, y, x)
\end{align}

\section{Implementation}
\label{sec:implementation}
For our implementation, we instantiate SLH with the theory of bit-vector arithmetic.  
Thus, according to Corollary~\ref{thm:decidability}, the resulting theory SLH[$\mathcal{T_{BV}}$] is decidable.
In this section, we provide details about the implementation of the decision procedure via a reduction to SAT.

To check validity of an SLH[$\mathcal{T_{BV}}$] formula $\phi$ we search for a small counterexample $H$.
By Theorem~\ref{thm:small-model}, if no such small $H$ exists, there is no counterexample
and so $\phi$ is a tautology.  We encode the existence of a small counterexample by constructing
a SAT formula.  

To generate the SAT formula, we instantiate every occurence of the
SLH[$\mathcal{T_{BV}}$] functions with the functions shown in
Figure~\ref{fig:operators}.  The structure that the functions operate over
is the following, where $N$ is the number of vertices in the structure and
$P$ is the number of program variables:

\begin{lstlisting}[mathescape=true]
typedef int node;
typedef int ptr;

struct heap {
  ptr: node[P];
  succ : (node $\times$ int)[N];
  nnodes: int;
}
\end{lstlisting}

The heap contains $N$ nodes, of which \lstinline|nnodes| are allocated. 
Pointer variables are represented as integers in the range $[0, P-1]$ where
by convention {\bf null} = 0.  Each pointer variable is mapped to an
allocated node by the \lstinline|ptr| array, with the restriction that {\bf
null} maps to node 0.  The edges in the graph are encoded in the
\lstinline|succ| array where \lstinline|h.succ[n] = (m, w)| iff the edge
$(n, m)$ with weight $w$ is in the graph.  For a heap with $N$ nodes, this
structure requires $3N + 1$ integers to encode.

The implementations of the SLH[$\mathcal{T_{BV}}$] functions described in
Section~\ref{sec:slh-syntax} are given in Figure~\ref{fig:operators}.  Note
that only \verb|Alloc| and \verb|Lookup| can allocate new nodes. 
Therefore if we are searching for a counterexample heap with at most $2P$
nodes, and our formula contains $k$ occurences of \verb|Alloc| and
\verb|Lookup|, the largest heap that can occur in the counterexample will
contain no more than $2P + k$ nodes.  We can therefore encode all of the
heaps using $6P + 3k + 1$ integers each.

When constructing the SAT formula corresponding to the SLH[$\mathcal{T_{BV}}$] formula, each of the functions
can be encoded (via symbolic execution) as a formula in the background theory $\mathcal{T_B}$ of constant size, 
except for \verb|PathLength|
which contains a loop.  This loop iterates $N = 2P + k$ times and so expands to a formula
of size $O(P)$.  If the SLH[$\mathcal{T_{BV}}$] formula contains $x$ operations, the final SAT formula in $\mathcal{T_{BV}}$ is therefore
of size $x \times P$.  We use {\sc CBMC} to construct and solve the SAT formula.

One important optimisation when constructing the SAT formula involves a
symmetry reduction on the counterexamples.  Since our encoding assigns names
to each of the vertices in the graph, we can have multiple representations
for heaps that are isomorphic.  To ensure that the SAT solver only considers
a single counterexample from each homeomorphism class, we choose a canonical
representative of each class and add a constraint that the counterexample we
are looking for must be one of these canonical representatives.  We define
the canonical form of a heap such that the nodes are ordered topologically
and so that the ordering is compatible with the ordering on the progam
variables:
\begin{align*}
 \forall p, p' \in P & . p < p' \Rightarrow \forall n, n' . L(p) \rightarrow^* n \wedge L(p') \rightarrow^* n' \Rightarrow n \leq n' \\
 \forall n, n' . & n \rightarrow n' \Rightarrow n \leq n' \\
 \intertext{Where $n \rightarrow^* n'$ means $n'$ is reachable from $n$.}
\end{align*}

\begin{figure}
\begin{multicols}{2}
\scriptsize
\begin{algorithmic}
\Function{NewNode}{heap h}
 \Let{n}{h.nnodes}
 \Let{h.nnodes}{h.nnodes + 1}
 \Let{h.succ[n]}{(null, 1)}
 \State \Return{n}
\EndFunction

\Statex

\Function{Subdivide}{heap h, node a}
\Let{n}{NewNode(h)}
\Let{(b, w)}{h.succ[a]}
\Let{h.succ[a]}{(n, 1)}
\Let{h.succ[n]}{(b, w - 1)}
\State \Return n
\EndFunction

\Statex

\Function{Update}{heap h, ptr x, ptr y}
 \Let{n}{h.ptr[x]}
 \Let{m}{h.ptr[y]}
 \Let{h.succ[n]}{(m, 1)}
\EndFunction

\Statex

\Function{Assign}{heap h, ptr x, ptr y}
\Let{h.ptr[x]}{h.ptr[y]}
\EndFunction

\Statex

\Function{Lookup}{heap h, ptr x, ptr y}
 \Let{n}{h.ptr[y]}
 \Let{(n', w)}{h.succ[n]}
 \If{w $\neq$ 1}
  \Let{n'}{Subdivide(h, n)}
 \EndIf
 \Let{h.ptr[x]}{n'}
\EndFunction

\columnbreak

\Function{Alloc}{heap h, ptr x}
\Let{n}{NewNode(h)}
\Let{h.ptr[x]}{n}
\EndFunction

\Statex

\Function{PathLength}{heap h, ptr x, ptr y}
\Let{n}{h.ptr[x]}
\Let{m}{h.ptr[y]}
\Let{distance}{0}

\For{i $\gets$ 0 to h.nnodes}
 \If{n = m}
  \State \Return{distance}
 \Else
  \Let{(n, w)}{h.succ[n]}
  \Let{distance}{distance + w}
 \EndIf
\EndFor

\State \Return{$\infty$}
\EndFunction

\Statex

\Function{Circular}{heap h, ptr x}
\Let{n}{h.ptr[x]}
\Let{m}{h.succ[n]}
\Let{distance}{0}

\For{i $\gets$ 0 to h.nnodes}
 \If{m = n}
  \State \Return{True}
 \Else
   \If{n = \hnull}
    \State \Return{False}
   \EndIf
 \EndIf

 \Let{m}{h.succ[m]}
\EndFor

\State \Return{False}
\EndFunction
\end{algorithmic}
\end{multicols}

\caption{Implementation of the SLH[$\mathcal{T_{BV}}$] functions\label{fig:operators}}
\end{figure}

\section{Motivation Revisited} \label{sec:revisit.motivation}

In this section, we get back to the motivational examples in
Figure~\ref{fig:motivation} and express their safety invariants and
termination arguments in SLH. As mentioned in Section~\ref{sec:prog_lang}, for ease of use, we don't mention the explicit heap $h$ in the specifications. 

In Figure~\ref{fig:motivation.a}, assuming that the call to the $length$
function ensures the state before the loop to be $\hpathlen(h, x, \nil)=n$,
then a possible safety invariant is $\hpathlen(h, y, \nil) = n-i$.  Note
that this invariant covers both the case where the list pointed by $x$ is
acyclic and the case where it contains a cycle.  In the latter scenario,
given that $\infty-i = \infty$, the invariant is equivalent to $\hpathlen(h,
y, \nil) = \infty$.  A ranking function for this program is $R(i)=-i$.

The program in Figure~\ref{fig:motivation.c} is safe with a possible safety
invariant:
%
$$\hpathlen(h, z, \nil)=\hpathlen(h, t, \nil).$$
Similar to the previous case, this invariant covers the scenario where the
lists pointed by $x$ and $y$ are acyclic, as well as the one where they are
cyclic.  In the latter situation, the program does not terminate.

  

For the example in Figure~\ref{fig:motivation.d},  the $divides$ function is
safe and a safety invariant is:
\begin{align*}
& \hispath(x, \hnull)  \wedge \hispath(z, \hnull)  \wedge \hispath(y, \hnull)  \wedge \hispath(y, z)  \wedge  \hispath(x, w)  \wedge\\ 
& \neg \hisnull(y)  \wedge (\hpathlen(x, w) + \hpathlen(z, \hnull))\%\hpathlen(y, \hnull) = 0. 
\end{align*}

Additionally, the function terminates as witnessed by the ranking function $R(w) = \hpathlen(w,\hnull)$.

Function $isCircular$ in Figure~\ref{fig:motivation.d} is safe and terminating with the safety invariant:
$\hpathlen(l, p) \wedge \hpathlen(p, q) \wedge \hispath(q, p) {\neq} \hispath(l, \nil)$, 
and lexicographic ranking function: $R(q, p) = (\hpathlen(q, \hnull), \hpathlen(q, p))$.






\section{Experiments}

To evaluate the applicability of our theory, we created a tool for verifying
that heaps don't lie: \shakira~\cite{hipsdontlie}.  We ran \shakira on a
collection of programs manipulating singly linked lists.  This collections
includes the standard operations of traversal, reversal, sorting etc.  as
well as the motivational examples from Section~\ref{sec:motivation}.  Each
of the programs in this collection is annotated with correctness assertions
and loop invariants, as well as the standard memory-safety checks.  One of
the programs (the motivational program from Figure~\ref{fig:motivation.c})
used a non-linear loop invariant, but this did not require any special
treatment by \shakira.

To generate VCs for each program, we generated a Hoare proof and then used
{\sc CBMC}~4.9~\cite{ckl2004} to compute the strongest post-conditions for
each Hoare triple using symbolic execution.  The resulting VCs were solved
using Glucose~4.0~\cite{glucose-paper}.  As well as correctness and memory
safety, these VCs proved that each loop annotation was genuinely a loop
invariant.  For four of the programs, we annotated loops with ranking
functions and generated VCs to check that the loops terminated, thereby
creating a total-correctness specification.

None of the proofs in our collection relied on assumptions about the shape
of the heap beyond that it consisted of singly linked lists.  In particular,
our safety proofs show that the safe programs are safe even in the presence
of arbitrary cycles and sharing between pointers.

We ran our experiments on a 4-core 3.30\,GHz Core i5 with 8\,GB of RAM.  The
results of these experiments are given in Table~\ref{tbl:experiments}.

\begin{table}
\centering
 \begin{tabular}{|l|c|c|c|c|c|c|}
 \hline
      & LOC & \#VCs & Symex(s) & SAT(s) & C/E \\
  \hline
  \hline
  \multicolumn{6}{|l|}{Safe benchmarks (UNSAT VCs)} \\
  \hline
  SLL (safe) & 236 & 40 & 18.2 & 5.9 & --- \\
  \hline
  SLL (termination) & 113 & 25 & 14.7 & 9.6 & --- \\
  \hline
 \hline
  \multicolumn{6}{|l|}{Counterexamples (SAT VCs)} \\
  \hline
  CLL (nonterm) & 38 & 14 & 6.9 & 1.6 & 3 \\
  \hline
  Null-deref & 165 & 31 & 13.6 & 3.0 & 3 \\
  \hline
  Assertion Failure & 73 & 11 & 3.5 & 0.7 & 3.5 \\
  \hline
  Inadequate Invariant & 37 & 4 & 4.9 & 1.2 & 6 \\
  \hline
 \end{tabular}

 \caption{Experimental results\label{tbl:experiments}}
 \caption*{
 Legend:

 \begin{tabular}{l l}
  LOC & Total lines of code \\
  \#VCs & Number of VCs \\
  Symex(s) & Total time spent in symbolic execution to generate VCs \\
  SAT(s) & Total time spent in SAT solver \\
  C/E & Average counterexample size (number of nodes)
 \end{tabular}}

\end{table}

The top half of the table shows the aggregate results for the benchmarks in
which the specifications held, i.e., the VCs were unsatisfiable.  These
``safe'' benchmarks are divided into two categories: partial- and
total-correctness proofs.  Note that the total-correctness proofs involve
solving more complex VCs -- the partial correctness proofs solved 40 VCs in
5.9\,s, while the total correctness proofs solved only 25 VCs in 9.6\,s. 
This is due to the presence of ranking functions in the total-correctness
proofs, which by necessity introduces a higher level of arithmetic
complexity.

The bottom half of the table contains the results for benchmarks in which
the VCs were satisfiable.  Since the VCs were generated from a Hoare proof,
their satisfiability only tells us that the purported proof is not in fact a
real proof of the program's correctness.  However, \shakira outputs models
when the VCs are satisfiable and these can be examined to diagnose the cause
of the proof's failure.  For our benchmarks, the counterexamples fell into
four categories:
\begin{itemize}
 \item Non-termination due to cyclic lists.
 \item Null dereferences.
 \item A correctness assertion (not a memory-safety assertion) failing.
 \item The loop invariant being inadequate, either by being too weak to prove the required properties,
  or failing to be inductive.
\end{itemize}

A counterexample generated by \shakira is shown in Figure~\ref{fig:experiments-cex}.
This program is a variation on the motivational program from Figure~\ref{fig:motivation.d}
in which the programmer has tried to speed up the loop by unwinding it once.  The result is
that the program no longer terminates if the list contains a cycle whose size is exactly one,
as shown in the counterexample found by \shakira.

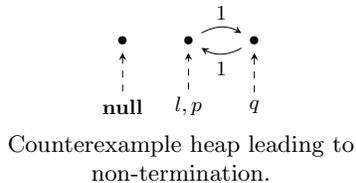
\begin{figure}
\centering
\begin{minipage}{.4\textwidth}
\centering
\begin{lstlisting}[basicstyle=\tiny,mathescape=true]
int has_cycle(list l) {
  list p = l;
  list q = l$\rightarrow$n;

  do {
    // Unwind loop to search
    // twice as fast!
    if (p != NULL) p = p$\rightarrow$n;
    if (p != NULL) p = p$\rightarrow$n;

    if (q != NULL) q = q$\rightarrow$n;
    if (q != NULL) q = q$\rightarrow$n;
    if (q != NULL) q = q$\rightarrow$n;
    if (q != NULL) q = q$\rightarrow$n;
  } while (p != q &&
           p != NULL &&
           q != NULL);

  return p == q;
}
\end{lstlisting}
\end{minipage}
\begin{minipage}{.45\textwidth}
\centering
\vfill
\resizebox{.45\textwidth}{!}{
\begin{tikzpicture}
 \node (n1) {$\bullet$};
 \node (n2) [right of=n1] {$\bullet$};
 \node (n3) [right of=n2] {$\bullet$};

 \node (null) [below of=n1] {\hnull};
 \node (l) [below of=n2] {$l, p$};
 \node (q) [below of=n3] {$q$};

 \path[->,draw] (n2) edge [bend left] node[above] {$1$} (n3)
                (n3) edge [bend left] node[below] {$1$} (n2);
                
 \path[->,dashed] (null) edge (n1)
                  (l) edge (n2)
                  (q) edge (n3);
\end{tikzpicture}
}

Counterexample heap leading to non-termination.
\vfill
\end{minipage}

\caption{A non-terminating program and the counterexample found by \shakira.\label{fig:experiments-cex}}
\end{figure}

These results show that discharging VCs written in SLH is practical with
current technology.  They further show that SLH is expressive enough to
specify safety, termination and correctness properties for difficult
programs.  When the VCs require arithmetic to be done on list lengths, as is
required when proving termination, the decision problem becomes noticeably
more difficult.  Our encoding is efficient enough that even when the VCs
contain non-linear arithmetic on path lengths, they can be solved quickly by
an off-the-shelf SAT solver.

\section{Related Works}

%

Research works on relating the shape of data structures to their numeric properties (e.g. length) 
follow several directions.
For abstract interpretation based analyses, an abstract domain that 
captures both heap and size was proposed in~\cite{DBLP:conf/vmcai/BouajjaniDES12}. 
The THOR tool \cite{DBLP:conf/cav/MagillTLT08,DBLP:conf/popl/MagillTLT10}
implements a separation logic \cite{DBLP:conf/lics/Reynolds02} based shape analysis and 
uses an off-the-shelf arithmetic analysis tool to
add support for arithmetic reasoning.
This approach is conceptually different from ours as it aims 
to separate the shape reasoning from the numeric
reasoning by constructing a numeric program that explicitly tracks changes
in data structure sizes.
In~\cite{DBLP:conf/atva/BouajjaniDES12}, Boujjani et al. introduce the logic
SLAD for reasoning about singly-linked lists and arrays with unbounded data, 
which allows to combine shape constraints, written in a fragment of separation logic, 
with data and size constraints.  
While SLAD is a powerful logic and has a decidable fragment, 
our main motivation for designing a new logic 
was its translation to SAT. A second motivation was the unrestricted sharing. 

%

%
Other recent decidable logics for reasoning about linked lists were
developed~\cite{DBLP:conf/cav/ItzhakyBINS13,DBLP:conf/cav/PiskacWZ13,DBLP:journals/jlp/YorshRSMB07,DBLP:conf/popl/MadhusudanPQ11,DBLP:conf/atva/BouajjaniDES12}. 
Piskac et al.~provide a reduction of decidable separation logic fragments
to a decidable first-order SMT theory~\cite{DBLP:conf/cav/PiskacWZ13}.
%
%
A decision procedure for an alternation-free
sub-fragment of first-order logic with transitive closure 
is described in~\cite{DBLP:conf/cav/ItzhakyBINS13}.
%
Lahiri and Qadeer introduce the Logic of Interpreted Sets and Bounded Quantification (LISBQ) capable to
express properties on the shape and data of composite data structures \cite{DBLP:conf/popl/LahiriQ08}.
In \cite{DBLP:conf/esop/BrainDKS14}, Brain et al. propose a decision procedure for
reasoning about aliasing and reachability based on Abstract Conflict Driven Clause Learning (ACDCL) \cite{DBLP:conf/popl/DSilvaHK13}. 
As they don't capture the lengths of lists, these logics are better suited for safety and less for termination proving.

%
%
%

\section{Conclusions}
We have presented the logic SLH for reasoning about potentially cyclic singly-linked lists. The main characteristics of SLH 
are the fact that it allows unrestricted sharing in the heap and can relate the structure of lists to their length, i.e. reachability constraints with numeric ones.
As SLH is parametrised by the background arithmetic theory used to express the length of lists,
we present its instantiation SLH[$\mathcal{T_{BV}}$] with the theory of bit-vector arithmetic and  provide a way of efficiently 
deciding its validity via a reduction to SAT. 
We empirically show that SLH is both efficient and expressive enough 
for reasoning about safety and (especially) termination of list programs.

\bibliography{all}{}
\bibliographystyle{splncs}

\end{document}